\documentclass[11pt]{article}
\usepackage{authblk,fullpage,times}    % not approved by ACM
\usepackage[utf8]{inputenc}
\usepackage{amsmath}
\usepackage{amssymb}

\usepackage{amstext}
\usepackage{amsfonts}
\usepackage{amsthm}
\usepackage{bbold}
\usepackage{mathtools}
\usepackage{color}
\usepackage{hyperref}
\usepackage{bm}
\usepackage{enumitem}

\newtheorem{theorem}{Theorem}[section]    % Specify Theorem
\newtheorem{definition}{Definition}[section] % Specify Definition
\newtheorem{corollary}[theorem]{Corollary}    % Specify Corollary
\newtheorem{lemma}[theorem]{Lemma}    % Specify Lemma

\renewcommand{\qed}{\hfill{$\rule{6pt}{6pt}$}} %Box at end of proof
\renewenvironment{proof}{\noindent{\bf Proof:}}{\qed\\}

\numberwithin{equation}{section}

\newcommand{\size}[1]{\left| #1 \right|}
\newcommand{\set}[1]{\left\{ #1 \right\}}

\newcommand{\ceil}[1]{{\lceil #1 \rceil}}

\DeclareMathOperator{\trace}{Tr}

\DeclareMathOperator{\Order}{O}

\DeclareMathOperator{\expct}{{\mathbb E}}

\newcommand{\ket}[1]{| #1 \rangle}

\newcommand{\ketbra}[2]{| #1 \rangle\!\langle #2 |}

\newcommand{\density}[1]{\ketbra{#1}{#1}}
\newcommand{\ancilla}{\ket{\bar{0}}}

\newcommand{\eqdef}{\coloneqq}
\newcommand{\tensor}{\otimes}

\newcommand{\union}{\cup}

\newcommand{\meet}{\wedge}

\newcommand{\adjoint}{*}

\newcommand{\suppress}[1]{}

\newcommand{\etal}{\emph{et al.\/}}

\DeclareMathOperator{\entropy}{S}
\DeclareMathOperator{\mi}{I}

\newcommand{\rA}{{\mathsf A}}

\newcommand{\cO}{{\mathcal O}}
\newcommand{\cH}{{\mathcal H}}
\newcommand{\cW}{{\mathcal W}}
\newcommand{\cA}{{\mathcal A}}
\newcommand{\linedisj}{{\mathsf L}}
\newcommand{\disj}{{\mathsf D}}
\DeclareMathOperator{\il}{IL}
\DeclareMathOperator{\til}{\widetilde{IL}}
\newcommand{\diam}{\delta}
\newcommand{\hX}{\hat{X}}
\newcommand{\hY}{\hat{Y}}
\newcommand{\hZ}{\hat{Z}}

\newcommand{\nodes}{p}

\title{Quantum Distributed Complexity of Set Disjointness on a Line~\footnote{A preliminary version of this article appeared in the proceedings of ICALP~2020~\cite{MN20-line-disjointness}. Among other improvements, this article corrects an error in the statement of a result from prior work (Theorem~\ref{thm-jrs}) in the conference version, and makes corresponding changes in the rest of the article.}}

\author{Fr{\'e}d{\'e}ric Magniez~\thanks{IRIF, 
Universit{\'e} de Paris and CNRS, 75205 Paris Cedex 13,
France.
Email: \texttt{magniez@irif.fr}~.
}}
\affil{Universit\'e de Paris, IRIF, CNRS, France}
\author{Ashwin Nayak~\thanks{Department of Combinatorics and Optimization,
and Institute for Quantum Computing, University
of Waterloo, 200 University Ave.\ W., Waterloo, ON,
N2L~3G1, Canada. Email: \texttt{ashwin.nayak@uwaterloo.ca}~.
}}
\affil{University of Waterloo, Canada}

\date{September 23, 2021}

\begin{document}

\maketitle

\begin{abstract}
Given~$x,y\in\{0,1\}^n$, Set Disjointness consists in deciding whether~$x_i=y_i=1$ for some index~$i \in [n]$. We study the problem of computing this function in a distributed computing scenario in which the inputs~$x$ and~$y$ are given to the processors at the two extremities of a path of length~$d$. Each vertex of the path has a quantum processor that can communicate with each of its neighbours by exchanging $\Order(\log n)$ qubits per round. We are interested in the number of rounds required for computing Set Disjointness with constant probability bounded away from~$1/2$. We call this problem ``Set Disjointness on a Line''. 

Set Disjointness on a Line was introduced by Le~Gall and Magniez~\cite{LM18-diameter-congest-model} for proving lower bounds on the quantum distributed complexity of computing the diameter of an arbitrary network in the CONGEST model. However, they were only able to provide a lower bound when the local memory used by the processors on the intermediate vertices of the path is severely limited. More precisely, their bound applies only when the local memory of each intermediate processor consists of $\Order( \log n)$ qubits.

In this work, we prove an unconditional lower bound of $\widetilde{\Omega}\big(\sqrt[3]{n d^2}+\sqrt{n} \, \big)$ rounds for Set Disjointness on a Line with~$d + 1$ processors. This is the first non-trivial lower bound when there is no restriction on the memory used by the processors. The result gives us a new lower bound of~$\widetilde{\Omega} \big( \sqrt[3]{n\delta^2}+\sqrt{n} \, \big)$ on the number of rounds required for computing the diameter~$\delta$ of any $n$-node network with quantum messages of size $\Order(\log n)$ in the CONGEST model.

We draw a connection between the distributed computing scenario above and a new model of query complexity. In this model, an algorithm computing a bi-variate function~$f$ (such as Set Disjointness) has access to the inputs~$x$ and~$y$ through two separate oracles~$\cO_x$ and~$\cO_y$, respectively. The restriction is that the algorithm is required to alternately make~$d$ queries to~$\cO_x$ and~$d$ queries to~$\cO_y$, with input-independent computation in between queries. The model reflects a ``switching delay'' of~$d$ queries between a ``round'' of queries to~$x$ and the following ``round'' of queries  to~$y$. The information-theoretic technique we use for deriving the round lower bound for Set Disjointness on a Line also applies to the number of rounds in this query model. We provide an algorithm for Set Disjointness in this query model with round complexity that matches the round lower bound stated above, up to a polylogarithmic factor. This presents a barrier for obtaining a better round lower bound for Set Disjointness on the Line. At the same time, it hints at the possibility of better communication protocols for the problem.
% In this sense, the round lower bound we show for Set Disjointness on a Line is optimal.
\end{abstract}

%\newpage

\section{Introduction}

\subsection{Context}

The field of Distributed Computing aims to model a collection of processors or computers communicating with each other over some network with the goal of collectively solving a global computational task. This task may depend on the structure of the network and on some additional data distributed among the computers. For instance, one may want to compute the distance between two nodes of the network, or its diameter, a proper colouring, a spanning tree, or even all-pairs shortest paths. 
%Nowadays, not only computation can be distributed but also data.
In the context of cloud computing, data centres serve as special nodes of the network where data are stored. These centres are usually spread all over the world in order to minimise access time by clients.
Since some operations need to be performed in order to synchronise the centres, the distance between these centres influence the quality of the network. For instance, one may want to decide if there is any inconsistency between two or more remote databases, or check for the availability of a common slot for booking some service.

In this work, we focus on the case of two remote data centres deployed on two nodes of a distributed network, and consider the problem of computing Set Disjointness. This fundamental problem, which we denote by~$\disj_n$, consists in deciding whether two $n$-bit input strings~$x$ and~$y$ modelling two remote databases have the bit~$1$ at the same position. (This may indicate a schedule conflict, for instance.)
The problem has been studied extensively in Communication Complexity~\cite{YaoSTOC79}, due to its many applications in other contexts (see, for example, the survey by Chattopadhyay and Pitassi~\cite{CP10-set-disjointness}). In the most basic setting, two remote parties, Alice and Bob, hold the inputs~$x$ and~$y$, respectively. They communicate with each other directly in order to solve the problem, while minimising the total length of the messages exchanged. Depending upon the model of computation and the type of communication channel connecting the players, the messages may be deterministic, randomised, or quantum.

The two-party communication model may be too simplistic in some scenarios, since it assumes instantaneous communication and full access to the input (by the party that is ``given'' the input). To address the first issue,
we may include the communication delay as a multiplicative factor in the communication complexity. However, this would not account for a potentially more sophisticated use of the communication channel between the two parties. Consider the case when the channel consists of a chain of~$d$ devices, say, repeaters. One could use the channel as a network of processors in order to minimise the communication delay, for instance using cached memories.
With regard to the second issue---pertaining to access to the input---the standard two-party model may not be suitable when the inputs are massive, and may only be accessed in small parts. Such access is better modelled as in Query Complexity, in which inputs are accessed by querying \emph{oracles\/} (see, e.g., Refs.~\cite{BuhrmanW02,deWolf19,A18-query-complexity}).
%and was mostly studied for one oracle access only. 

Motivated by a concrete problem in distributed computing, we define a new model of query complexity, two-oracle query complexity with a ``switching delay''.  In this model we consider a single computer with access to two oracles, one for each input~$x$ or~$y$, such that switching between queries to the two inputs involves a time delay~$d$. The delay accounts for the lag in communication between the parties holding the inputs, for instance when the inputs are not physically at the same place. It might be advantageous to balance this delay by making several accesses to the same input, say~$x$, before switching to the other input~$y$; we also incorporate this feature in the model. The new model attempts to address both the issues discussed above, and is described more precisely in Section~\ref{sec-d-query-model}. 

There are several bridges between query complexity and communication complexity, but we are not aware of any previous work in a query model such as the one above.
The two models---communication through a chain of~$d$ devices, and two-oracle query algorithms with a switching delay of~$d$---share some similarities but also have fundamental differences. In the first model, one node has full access to half of the input. In the second model, all the information obtained so far from the inputs~$x$ and~$y$ is kept in the same memory registers, even when the algorithm switches between inputs.

In this work, we show that the above refinements of the two-party communication model and the query model differ significantly from their standard versions for solving Set Disjointness in the quantum setting. Such a difference does not occur in the setting of deterministic or randomised computing, and we do not know whether such a difference arises for another ``natural'' problem.

\subsection{Application to quantum distributed computing}

This study was initially motivated by a problem left open by Le~Gall and Magniez~\cite{LM18-diameter-congest-model} in the context of distributed computing with congestion (CONGEST model). 
They demonstrated the superiority of quantum computing for computing the diameter~$\diam$ of a network with~$\nodes$ nodes (Diameter problem). They designed a quantum distributed algorithm using $\widetilde{\Order}(\sqrt{\nodes \diam} \,)$ synchronised rounds, where simultaneous messages of $\Order(\log \nodes)$ qubits are exchanged at each round between neighbouring nodes in the network. They also established a lower bound of $\widetilde{\Omega}(\sqrt{\nodes}+\diam)$ rounds. 

Classically the congested distributed complexity of Diameter is well understood, and requires $\widetilde{\Theta}(\nodes)$ rounds~\cite{Holzer+PODC12,Peleg+ICALP12,Frischknecht+SODA12}. 
The lower bound is based on the construction of a two-party communication protocol for Set Disjointness from any distributed algorithm for Diameter.
From $n$-bit inputs $x,y$, two pieces of a $\Theta(n)$-node network are constructed by the two players. Then the pieces are connected by~$\Order(\log n)$ edges. The diameter of the resulting network is either $4$ or $5$ depending on the solution to Set Disjointness with inputs~$(x,y)$.
%by disposing the two remote players $A$ and $B$ on two nodes of some elaborated $n$-node network made of cleaver graph gadgets. 
In each round of a distributed algorithm for solving the Diameter Problem on the graph, the total communication between the two pieces of the network is~$\Order( \log^2 n)$. So the classical lower bound of~$\Omega(n)$ for the communication complexity of Set Disjointness implies an~$\Omega(n/ \log^2 n ) $ lower bound on the number of rounds used by the algorithm.

In the quantum setting, the situation is much more complex since Set Disjointness has communication complexity~$\Theta(\sqrt{n} \,)$ for~$n$-bit inputs~\cite{Razborov03-set-disjointness,AA03-spatial-search}. This leads to the lower bound of~$\widetilde{\Omega}(\sqrt{\nodes}+\diam)$ rounds for computing the diameter of a quantum congested network, which is significantly smaller than the upper bound stated above. Nonetheless, Le~Gall and Magniez improved the lower bound for a restricted set of protocols in which each node has memory of size at most~$\mathrm{poly}(\log \nodes)$ qubits.
For this, they used a more refined lower bound for Set Disjointness for bounded-round protocols.

Recall that the number of rounds in a two-party protocol is the number of messages exchanged, where the length of the messages may vary. Braverman, Garg, Ko, Mao, and Touchette~\cite{BGKMT18-bounded-round-disj} showed that the communication complexity of~$r$-round two-party quantum protocols for Set Disjointness on~$n$-bit inputs is~$\widetilde{\Omega}(n/r+r)$. Using this, Le~Gall and Magniez showed that any quantum distributed protocol for Diameter with congestion $\Order(\log \nodes)$ and memory-size $\mathrm{poly}(\log \nodes)$ per node
requires $\widetilde{\Omega}(\sqrt{\nodes \diam} \,)$ rounds. However, without any restriction on the memory size of the nodes, no better bound than $\widetilde{\Omega}(\sqrt{\nodes}+\diam)$ was known.
%one can argue that this memory is then bounded by the number $n$ of nodes times the congestion, leading to a lower bound of $\Tilde{\Omega}(\sqrt[3]{nD} \,)$ for congestion of $\mathrm{poly}(\log n)$. This bound is useless since it is lower than $\sqrt{n}$ when $D\leq n$, which is always the case.

\subsection{Contributions}

We prove that solving Set Disjointness with the two~$n$-bit inputs given to the processors at the extremities of a line of~$d+1$ quantum processors requires~$\widetilde{\Omega}(\sqrt[3]{n d^2} \,)$ rounds of communication of messages of size $\Order(\log n)$ (\textbf{Theorem~\ref{thm-lb}}).
As a corollary, we get a new lower bound of $\widetilde{\Omega}(\sqrt[3]{\nodes \diam^2} \,)$ rounds for quantum distributed protocols computing the diameter~$\diam$ of a~$\nodes$-node network with congestion $\Order(\log \nodes)$ (\textbf{Corollary~\ref{cor-diam}}). This bound improves on the previous bound of $\widetilde{\Omega}(\sqrt{\nodes} \,)$ rounds when~$\diam \in \widetilde{\Omega}(\sqrt[4]{\nodes} \,)$. The improvement is obtained by a more refined, information-theoretic analysis of a reduction similar to one due to Le~Gall and Magniez~\cite{LM18-diameter-congest-model}. 

We observe that the information-theoretic technique used to derive the above round lower bound for Set Disjointness on the Line also applies to two-oracle query algorithms with switching delay~$d$ (\textbf{Theorem~\ref{thm-lb2}}). We show that this bound, and the bound of ${\Omega}(\sqrt{n} \,)$ coming from the standard query complexity model, are tight to within polylogarithmic factors in different ranges of the parameters~$n$ and~$d$ (\textbf{Theorem~\ref{thm-algo}}). This presents a barrier for obtaining a better round lower bound for Set Disjointness on the Line.  At the same time, it hints at the possibility of better communication protocols for the problem. (Note that the complexity of the problem may be strictly in between the best known bounds.) We hope that these results and, more generally, the models we study also provide a better understanding of quantum distributed computing.

\paragraph{Acknowledgements.}
We are grateful to the anonymous referees for their thorough and helpful feedback. F.M.'s research is supported in part by the ERA-NET Cofund in Quantum Technologies project QuantAlgo and the French ANR Blanc project RDAM. A.N.'s research is supported in part by NSERC Canada.

\section{Preliminaries}

We assume that the reader is familiar with the basic notions of quantum information and computation. We recommend the texts by Nielsen and Chuang~\cite{NC00} and Watrous~\cite{W18-TQI}, and the lecture notes by de Wolf~\cite{deWolf19} for a good introduction to these topics. We briefly describe some notation, conventions, and the main concepts that we encounter in this work. 

We write pure quantum states using the ket notation, for example as~$\ket{\psi}$.
By a quantum register, we mean a sequence of quantum bits (qubits).
%When the system has several identified registers, we indicate them as a subscript. For instance $\ket{i}_\mathsf{R}\ket{j}_\mathsf{Q}$, means that register $\mathsf{R}$ is in state $\ket{i}$, and register $\mathsf{Q}$ in state~$\ket{j}$. Of course those could be entangled and in superposition. Then we  write 
%$\ket{\psi}=\sum_{i,j}\alpha_{ij}\ket{i}_\mathsf{R}\ket{j}_\mathsf{Q}$.
We assume for simplicity (and without loss of generality) that the computation in the models we study do not involve any intermediate measurements, i.e., they are unitary until the measurement that is made to obtain the output.

%We will use the phrase \emph{with high probability} to mean probability at least $1-\frac{1}{\mathrm{poly}}$ for some super-linear polynomial. We ensure that all our subroutines succeed with high probability, to achieve a bounded-error algorithm at the end. To achieve such high probability, we will necessarily incur $\mathrm{polylog}$ factors. 
We use the notation $\widetilde{\Order}( \cdot )$ to indicate that we are suppressing factors that are poly-logarithmic in the stated expression. For a positive integer $k$, we denote the set $\{1,2, \dotsc ,k\}$ by $[k]$.
In the sequel, we % let $X$, $Y$ and $Z$ be three finite sets. 
consider the computation of Boolean bi-variate functions~$f\colon \{0,1\}^n \times \{0,1\}^n\to \{0,1\}$ in several models of computation.
%that we want to compute within different models of computation.

\subsection{Quantum distributed computing in the CONGEST model}
\label{sec-congest}

We consider the quantum analogue of the standard CONGEST communication model. We give a brief overview here, and refer the reader to Ref.~\cite{Peleg00} for a more detailed discussion of the model and its variants. The topology of the network is given by some graph~$G \eqdef (V, E)$. Each node in the network has a distinct identifier and represents a processor. Initially, the nodes know nothing about the topology of the network except the set of edges incident on them, and a polynomial upper bound~$\Order(\size{V}^c)$ (for some constant~$c$) on the total number of nodes~$|V|$.  

There are a number of subtleties in the use of shared entanglement in this model, such as what shared states are allowed, how they are distributed, and what knowledge the processors have about the states. These considerations gain more importance in the design of algorithms in the model. That said, we are concerned with \emph{lower bounds\/} for distributed algorithms, and we prove them in a model in which the processors are \emph{the most powerful\/}. The lower bounds so obtained are thus stronger. 
We assume that the processors initially share an arbitrary entangled \emph{pure\/} state that depends only on the number of nodes (but not on the topology of the graph, nor on the inputs the processors may be given). Further, each processor knows the shared state and how it is partitioned amongst the processors in the network. Thus each processor initially also knows the precise number of nodes~$\size{V}$, but not the set of communication links beyond those with its neighbours.

Communication protocols in the CONGEST model are executed  with round-based synchrony. In each round, each node may perform some quantum computation on its local memory and the message registers it uses to communicate with its neighbours. Then each node transfers one message with~$b$ qubits to each adjacent node to complete that round. The parameter~$b$ is called the \emph{congestion\/} or \emph{bandwidth\/} of the communication channels. Unless explicitly mentioned, we assume that the congestion~$b$ is of order~$\log |V|$.
All links and nodes in the network (corresponding to the edges and vertices of $G$, respectively) are reliable and do not suffer any faults. 

In this paper we consider the special case of a~$d$-line network, where~$G$ consists in a single path of length~$d$. The nodes/processors at the extremities receive inputs~$x,y\in \{0,1\}^n$, respectively, and the intermediate nodes get no input. The~$d+1$ processors also share an arbitrary quantum state as described above. 
In this setting, the quantum distributed complexity of~$f$ on a $d$-line is the minimum number of rounds of any quantum protocol that computes $f$ with probability at least~$2/3$ and congestion $\Order(\log n)$.
The complexity of any non-trivial function~$f$ of both its arguments is~$\Omega(d)$. We assume that $d\leq n$; otherwise the complexity of such a function would be~$\Theta(d)$. Note that even in the model without entanglement, we may assume that~$d$ is known to each node. Otherwise,~$d$ can be computed at the cost of $\Theta(d)$ rounds, which does not affect the asymptotic complexity of such functions.

%This would be explained carefully in Section~\ref{sec-line}.

\subsection{Quantum information theory}

We refer the reader to the texts by Nielsen and Chuang~\cite{NC00} and Watrous~\cite{W18-TQI} for the basic elements of quantum information theory.

Unless specified, we take the base of the logarithm function to be~$2$. Whenever we consider information-theoretic quantities involving quantum registers, we assume they are in a quantum state that is implied by the context. For ease of notation, we identify the register with the quantum state. 

For a register~$X$ in state~$\rho$ the \emph{von Neumann entropy\/} of~$X$ is defined as~$\entropy(X)_\rho \coloneqq - \trace ( \rho \log \rho)$. We omit the subscript~$\rho$ when the state of the register is clear from the context. If the state space of~$X$ has dimension~$k$, then~$\entropy(X) \le \log k$. 

Suppose that the registers~$WXYZ$ are in some joint quantum state~$\rho$. The \emph{mutual information\/}~$\mi(X:Y)_\rho$ of~$X$ and~$Y$ is defined as
%~$\mi(X:Y) \eqdef \entropy(X) + \entropy(Y) - \entropy(XY)$.
\[
\mi(X:Y)_\rho \quad \eqdef \quad \entropy(X) + \entropy(Y) - \entropy(XY) \enspace.
\]
The \emph{conditional mutual information\/}~$\mi(X : Y \,|\, Z)_\rho$ of~$X$ and~$Y$ given~$Z$ is defined as
%~$\mi(X : Y \,|\, Z) \eqdef \mi(X : YZ) - \mi( X : Z)$.
\[
\mi(X : Y \,|\, Z)_\rho \quad \eqdef \quad \mi(X : YZ)_\rho - \mi( X : Z)_\rho \enspace.
\]
We omit the subscript~$\rho$ when the state of the registers is clear from the context. 

Conditional mutual information is invariant under the application of an isometry to any of its three arguments. The quantity also satisfies the following important property.
\begin{lemma}[Data Processing Inequality]
\label{lem-dpi}
$\mi( X : WY \,|\, Z) ~\ge~ \mi( X : Y \,|\, Z)$.
\end{lemma}
We may bound conditional mutual information as follows.
\begin{lemma}
\label{lem-mi-increase}
$\mi( X : WY \,|\, Z) ~\le~ 2 \entropy(W) + \mi( X : Y \,|\, Z)$.
\end{lemma}
The quantity simplifies if the register on which we condition is ``classical''.
\begin{lemma}
\label{lem-cq}
Let~$\sigma$ be a possible state of the registers~$XYZ$ given by
\[
\sigma \quad \eqdef \quad \sum_z \lambda_z \; \sigma_z^{XY} \tensor \density{z}^Z \enspace,
\]
where~$(\ket{z})$ is an orthonormal basis of the state space of register~$Z$, $\lambda$ is a probability distribution on this basis, and~$(\sigma_z)$ are possible states of the registers~$XY$. Then
\[
\mi(X : Y \,|\, Z)_\sigma \quad = \quad \expct_{z \sim \lambda} \mi(X : Y)_{\sigma_z}  \enspace.
\]
\end{lemma}

\subsection{Quantum communication complexity}
\label{sec-qcc}

We informally describe a two-party quantum communication protocol \emph{with shared entanglement\/} (also called an \emph{entanglement-assisted\/} two-party quantum communication protocol) for computing a bi-variate Boolean function~$f(x,y)$ of~$n$-bit inputs~$x,y$.
For a formal definition, we refer the reader to an article by Touchette~\cite{Touchette15-QIC}. In such a protocol, we have two parties,
Alice and Bob, each of whom gets an input in registers~$X$ and~$Y$, respectively. In the protocols we consider, the inputs are \emph{classical\/}, i.e., the joint quantum state in the input registers~$XY$ is diagonal in the basis~$( \ket{x,y} ~:~ x, y \in \set{0,1}^n )$. 
Alice and Bob's goal is to compute the value of the function on the pair of strings in the input registers by interacting with each other.

The protocol proceeds in some number~$m \ge 1$ of \emph{rounds\/}. At the cost of increasing the number of rounds by~$1$, we assume that Alice sends the message in the first round, after which the parties alternate in sending messages. Each party holds a \emph{work\/} register in addition to the input register. Initially, Alice has work register~$A_0 $, Bob has~$B_0$. We denote the work register with Alice at the end of round~$k \in [m]$ by~$A_{k}$ and that with Bob by~$B_{k}$.

The qubits in the work registers~$A_0 B_0$ are initialised to a fixed pure state that may be entangled across the partition across~$A_0$ and~$B_0$, but is independent of the inputs~$x,y$. This is called \emph{shared entanglement\/}.
(In the model \emph{without\/} shared entanglement, the registers~$A_0 B_0$ are initialised to~$\ancilla$.)
Suppose that Alice is supposed to send the message in the~$k$-th round, for some~$k \in [m]$. Alice applies an isometry controlled by her input register~$X$ to the work register~$A_{k - 1}$ to obtain registers~$A_k M_k$. She then sends the message register~$M_k$ to Bob. Bob's work register at the end of the~$k$-th round is then~$B_k \eqdef M_k B_{k-1}$. After the~$m$-th round (the last round), the recipient of the last message, say Bob, measures his work register~$B_k$, possibly controlled by his input register~$Y$, to produce the binary output of the protocol.

The length of a message is the number of qubits in the message register for that round. The \emph{entanglement-assisted communication complexity\/} of the protocol is the sum of the lengths of the~$m$ messages in it. We say the protocol computes the function~$f$ with success probability~$\alpha$ if for all inputs~$x,y$, the probability that the protocol outputs~$f(x,y)$ is at least~$\alpha$. The goal of the two parties is to compute the function while minimising the communication between themselves. The \emph{entanglement-assisted quantum communication complexity\/} of~$f$ is the minimum communication complexity of a quantum protocol with shared entanglement that computes~$f$ with success probability at least $2/3$.

We analyse a subtle variant of the \emph{conditional information loss\/} of two-party protocols, a notion introduced by Jain, Radhakrishnan, and Sen~\cite{JRS03-set-disjointness} (see Appendix~\ref{sec-jrs}). We call the variant \emph{conditional information leakage\/} to distinguish it from conditional information loss. This variant is implicit in Ref.~\cite{JRS03-set-disjointness}, and turns out to be the quantity of interest for us. We define this notion following the convention and notation given above. In particular, we assume that Alice sends the messages in the odd rounds and Bob sends the messages in the even rounds. Moreover, the only measurement in the protocol is the one for producing the output, so that the joint state of the two parties is pure for any fixed pair of inputs. Let~$\mu$ be a joint distribution over the input set~$\set{0,1}^n \times \set{0,1}^n$ and an auxiliary sample space. We initialise registers~$\hX \hY \hZ X Y Z $ to the canonical purification
\[
\sum_{x,y,z}  \sqrt{\mu(x,y,z)} \; \ket{xyz}^{\hX \hY \hZ} \ket{xyz}^{X Y Z}
\]
of the distribution, where~$Z$ corresponds to the auxiliary random variable.
We use the register labels~$X, Y, Z$ to also refer to the two input random variables~($X,Y$) and the auxiliary random variable~($Z$), respectively. We then run the two-party protocol~$\Pi$ using the input registers~$X,Y$ respectively, and additional work registers as described above. We imagine that the purification register~$\hX$ is given to Alice, the register~$\hY$ is given to Bob, and that the registers~$\hZ Z$ are held by a third party.

The \emph{conditional information leakage\/}~$\til( \Pi \,|\, XYZ)$ of the protocol~$\Pi$ is defined as
\[
\til( \Pi \,|\, XYZ) \quad \eqdef \quad \sum_{i \in [m], ~ i \text{ odd}} \mi( X : B_i Y \hY \,|\, Z) 
    + \sum_{i \in [m], ~ i \text{ even}} \mi( Y : A_i X \hX \,|\, Z) \enspace,
\]
where the registers are implicitly assumed to be in the state given by the protocol. Since Alice sends the messages in the odd rounds, and Bob in the even rounds, this quantity measures the cumulative information about the inputs ``leaked'' to the other party, over the course of the entire protocol.

\subsection{Quantum query complexity}
\label{sec-query-complexity}

For a thorough introduction to the quantum query model, see, for example, the lecture notes by de Wolf~\cite{deWolf19} and the survey by Ambainis~\cite{A18-query-complexity}. In this work, we study algorithms for computing a bi-variate Boolean function~$f(x,y)$ as above, using \emph{two\/} unitary operators~$\cO_x$ and~$\cO_y$ that provide access to the~$n$-bit inputs~$x$ and~$y$, respectively. For any~$z \in \set{0,1}^n$, the operator~$\mathcal{O}_z$ acts as~$\mathcal{O}_z \ket{i,b} = \ket{i,b\oplus z_i}$ on the Hilbert space~$\cH$ spanned by the orthonormal basis~$\{\ket{i,b}:i\in[n],b\in\{0,1\}\}$.
%\begin{align*}
%\mathcal{O}_z : \ket{i,b}\mapsto\ket{i,b\oplus z_i} \enspace.
%\end{align*}
We call operators of the form~$\mathcal{O}_z$ an \emph{oracle\/}, and each application of such an operator a \emph{query\/}.

A query algorithm~$\cA$ with access to two oracles~$\cO_x$ and~$\cO_y$ is an alternating sequence of unitary operators~$U_0, V_1, U_1, V_2, U_2, V_3, U_3, \dotsc, V_t, U_t$, where the operators~$U_i$ act on a Hilbert space of the form~$\cH \tensor \cW$ and are independent of the inputs~$x,y$, and the~$V_i \in \set{\cO_x,\cO_y}$. The computation starts in a fixed state~$\ancilla \in \cH \tensor \cW$, followed by the sequence of unitary operators to get the final state~$U_t V_t \dotsb U_3 V_3 U_2 V_2 U_1 V_1 U_0 \ancilla$. Finally, we measure the first qubit in the standard basis to obtain the output~$\cA(x,y)$ of the algorithm. We say the algorithm computes~$f$ with success probability~$\alpha$ if for all inputs~$x,y$, we have~$\cA(x,y) = f(x,y)$ with probability at least~$\alpha$. 

As in the standard quantum query model, we focus on the number of applications of the operators~$\mathcal{O}_x$ and~$\mathcal{O}_y$ in an algorithm, and ignore the cost of implementing unitary operators that are independent of~$x$ and~$y$. The \emph{query complexity\/} of an algorithm is the number of queries made by the algorithm ($t$ in the definition above). The \emph{quantum query complexity\/} of a function~$f$ is the minimum query complexity of any quantum algorithm that computes~$f$ with probability at least~$2/3$.

\section{Set Disjointness on a Line}
\label{sec-linedisj}

\subsection{The problem and results}

The Set Disjointness problem~$\linedisj_{n,d}$ on a line was introduced recently 
by Le Gall and Magniez~\cite{LM18-diameter-congest-model} in the context 
of distributed computing.  It is a communication 
problem involving~$d + 1$ communicating parties, $\rA_0, \rA_1, \dotsc,
\rA_d$, arranged on the vertices of a path of length~$d$. The edges of
the path denote two-way quantum communication channels between the
players.
Parties~$\rA_0$ and~$\rA_d$ receive~$n$-bit inputs~$x, y \in
\set{0,1}^n$, respectively. The~$ d + 1$ parties also share an arbitrary 
entangled state as described in Section~\ref{sec-congest}.
The communication protocol proceeds in
rounds. In each round, parties~$\rA_{i - 1}$ and~$\rA_i$ may 
exchange~$b$ qubits in each direction, for each~$i \in [d]$, i.e., the 
\emph{bandwidth\/} of each communication channel is~$b$. The goal of
the parties is to determine if the sets~$x$ and~$y$ intersect or not.
I.e., they would like to compute the Set Disjointness function~$\disj_n(x,y) 
\eqdef \bigvee_{i = 1}^n (x_i \meet y_i)$. 

We are interested in the number of rounds required to solve~$\linedisj_{n,d}$.
We readily get a quantum protocol~$\Pi_d$ for this problem 
with~$\Order( \sqrt{nd} \,)$ rounds by following an observation due to
Zalka~\cite{Zalka99-unordered-search} on black-box algorithms that make 
``parallel'' queries. Let~$\Pi$ denote the optimal two-party
quantum communication protocol for Set Disjointness due to
Aaronson and Ambainis~\cite{AA03-spatial-search}.
In~$\Pi_d$, we partition the~$n$-bit inputs into~$d$ parts 
of length~$n/d$ each. Parties~$\rA_0$ and~$\rA_d$ then simulate~$\Pi$
on each of the~$d$ corresponding pairs of inputs independently.
The protocol~$\Pi$ runs in~$\sqrt{n/d}$ rounds with~$\Order(1)$ 
qubits of communication per instance of length~$n/d$, per round. So the total 
communication to or from~$\rA_0$ due to one round of the~$d$ runs
of~$\Pi$ is~$\Order(d)$. Since~$\Order(d)$ qubits can be 
transmitted across the path of length~$d$ in~$\Order(d)$ rounds of the 
multi-party protocol, the protocol~$\Pi_d$
simulates the~$d$ parallel runs of~$\Pi$ in~$\Order( \sqrt{nd} \,)$
rounds. Since~$\Pi$ finds an intersection with probability at
least~$3/4$ whenever there is one, and does not err when there is no 
intersection, the protocol~$\Pi_d$ also has the same correctness
probability.

Le Gall and Magniez observed that a lower bound of~$\Omega(
\sqrt{n}/ b )$ for the number of rounds follows from the~$\Omega( \sqrt{n} \,)$
lower bound due to Razborov~\cite{Razborov03-set-disjointness} on 
the quantum communication complexity of Set Disjointness in the 
two-party communication model. This is because 
two parties, Alice and Bob, may use any~$r$-round protocol for~$\linedisj_{n,d}$
to solve Set Disjointness with~$2 r b$ qubits of communication: Alice 
simulates~$\rA_0$ and Bob simulates the actions of the remaining 
parties~$\rA_1, \rA_2, \dotsc, \rA_d $.  An~$\Omega( d)$ lower bound is 
also immediate due to the need for communication between~$\rA_0$ and~$\rA_d
$.

Le Gall and Magniez devised a more intricate simulation of a
protocol for~$\linedisj_{n,d}$ by two parties, thereby obtaining a two-party protocol
for Set Disjointness. Using this, they obtained a round lower bound 
of~$\widetilde{\Omega}( \sqrt{nd} \,)$ for~$\linedisj_{n,d}$ when the 
bandwidth~$b$ of each communication channel (in each round) and
the local memory of the players~$\rA_1, \rA_2, \dotsc, \rA_{d-1}$ are
both~$\Order( \log n)$ qubits. We show that a similar simulation leads to an
unconditional round lower bound of~$\Omega(n d^2 / b)^{1/3}$ by studying the
conditional information leakage of the resulting two-party protocol (see Section~\ref{sec-qcc}).

\begin{theorem}
\label{thm-lb}
Any entanglement-assisted quantum communication protocol with error probability at most~$1/3$
for the Set Disjointness problem~$\linedisj_{n,d}$ on the line 
requires~$\Omega ( \sqrt[3]{n d^2 / b} \,)$ rounds.
\end{theorem}
This bound dominates the straightforward bound of~$\Omega( \sqrt{n}/ b)$
mentioned above when~$d \ge \sqrt[4]{n} / b$, i.e., when~$d \ge 
\sqrt[4]{n} / \log n$ when~$b \eqdef \log n$. However, we do not know if either
bound is achievable in the respective parameter regimes. We study the
optimality of the above bound via a related query model in Section~\ref{sec-d-query-algorithms}.

Using the reduction from~$\linedisj_{n,d}$ to the problem of computing the diameter
described in the proof of Theorem~1.3 in Ref.~\cite{LM18-diameter-congest-model}, 
we get a new lower bound for quantum distributed protocols for the diameter problem
in the CONGEST model. 
\begin{corollary}\label{cor-diam}
Any distributed protocol for computing the diameter~$\diam$ of~$p$-node 
networks with congestion $\Order(\log p)$ in the quantum CONGEST model (possibly with 
shared entanglement), requires  $\widetilde{\Omega}(\sqrt[3]{p \diam^2} \,)$ rounds.
\end{corollary}

\subsection{Overview of the proofs}
\label{sec-overview}

We begin by giving an overview of the proof of Theorem~\ref{thm-lb}. It rests on a simulation of a protocol for~$\linedisj_{n,d}$ by a two-party protocol for Set Disjointness similar to one designed by Le Gall and Magniez~\cite[Theorem~6.1]{LM18-diameter-congest-model}.
(In fact, the simulation works for any multi-party protocol over the path of length~$d$ that computes 
some bi-variate function~$g(x,y)$ of the inputs given to~$\rA_0$ and~$\rA_d$.)
The idea underlying the  simulation is the following.
Suppose we have a protocol~$\Pi_d$ for the problem~$\linedisj_{n,d}$.
In the two-party protocol~$\Pi$, Alice begins by holding the
registers used by parties~$\rA_0, \rA_1, \dotsc, \rA_{d - 1}$.
She then simulates all the actions---local operations and
communication---of the parties~$\rA_0, \rA_1, \dotsc, \rA_{d - 1}$ 
from the first round in~$\Pi_d $, except for the communication 
between~$\rA_{d - 1}$ and~$\rA_d $. This is possible because these actions 
do not depend on the input~$y$ held by~$\rA_d $. She can continue simulating the actions 
of~$\rA_0, \rA_1,
\dotsc, \rA_{d - 2}$ from the second round, except the communication
between~$\rA_{d - 2}$ and~$\rA_{d - 1} $, as these do not depend on the
message from~$\rA_d$ from the first round in~$\Pi_d $. Continuing this way,
Alice can simulate the actions of~$\rA_0, \rA_1, \dotsc,
\rA_{d - i}$ from round~$i$ of~$\Pi_d $, except the communication
between~$\rA_{d - i}$ and~$\rA_{d - i + 1} $, for all~$i \in [d]$, all 
in one round of~$\Pi$. These actions constitute Alice's local operations
in the first round of~$\Pi$.

Alice then sends Bob the local memory used by 
parties~$\rA_1, \dotsc, 
\rA_{d - 1}$ in~$\Pi_d $, along with the qubits sent by~$\rA_{i-1}$
to~$\rA_i$ in round~$i$, for each~$i \in [d]$. (Alice retains the 
input~$x$ and the memory used by party~$\rA_0 $.) This constitutes the 
first message from Alice to Bob in~$\Pi$.

Given the first message, Bob can simulate the remaining 
actions of~$\rA_1, \rA_2, \dotsc, \rA_d$ from the first~$d$ rounds 
of~$\Pi_d$, except for the communication from~$\rA_1$ to~$\rA_0$.
These constitute his local operations in the second round of~$\Pi$.
He then sends Alice the qubits sent by~$\rA_1$ to~$\rA_0$ in round~$d$
of~$\Pi_d$ along with the local memory used by the parties~$\rA_i $,
for~$i \in [d-1]$. (Bob retains the input~$y$ and the local memory used
by party~$\rA_d$.) This constitutes the second message in~$\Pi$.

In effect, the simulation implements the first~$d$ rounds of~$\Pi_d$ 
in two rounds of~$\Pi$ (see Figure~\ref{fig-simulation}).
The same idea allows Alice and Bob to simulate
the rest of the protocol~$\Pi_d$ while implementing each successive block
of~$d$ rounds of~$\Pi_d$ in two rounds of~$\Pi$, with communication per
round of the order of~$d(b + s)$, where~$b$ is the bandwidth of the
communication channels in~$\Pi_d $, and~$s$ is a bound on the number of
qubits of local memory used by any of the parties~$\rA_1, \dotsc, \rA_{d-1}$.
Building on the detailed description of protocols on the line in 
Section~\ref{sec-line-protocols}, we describe the simulation formally in 
Section~\ref{sec-simulation}, and show the following.
\begin{lemma}
\label{lem-simulation}
Given any~$r$-round entanglement-assisted quantum protocol~$\Pi_d$ for~$\linedisj_{n,d}$ over 
communication channels with bandwidth~$b$ in which each party uses local 
memory at most~$s$, there is an entanglement-assisted two-party quantum protocol~$\Pi$ for Set 
Disjointness~$\disj_n$ that has~$2 \ceil{r/d}$ rounds, total 
communication of order~$r(b + s)$, and has the same probability of success. Further,
if the protocol~$\Pi_d$ does not use shared entanglement, the protocol~$\Pi$ also does
not.
\end{lemma}

The communication required by a~$k$-round bounded-error two-party
protocol for Set Disjointness is~$\Omega( n/(k \log^8 k) )
$~\cite[Theorem~A]{BGKMT18-bounded-round-disj}. This gives us the
lower bound of~$\widetilde{\Omega}( \sqrt{nd} \,)$
due to Le Gall and Magniez on the number of rounds~$r$ in~$\Pi_d$, 
when~$b + s$ is of order~$\log n$.
More precisely, the bound with the logarithmic factors is
\[
\Omega \left( \frac{ \sqrt{nd}}{ (\log n)^{1/2} \log^4 \frac{n}{ d \log n} }
\right) \enspace.
\]

In fact, in the case the protocol~$\Pi_d$ does not use shared entanglement,
we may derive an unconditional lower bound on the number of
rounds from the same reduction, one that holds without any
restriction on the local memory used by the parties in~$\Pi_d$. This is
because the state of the registers of any party~$\rA_i$, with~$i \in [d-1]$, 
in an~$r$-round protocol without entanglement has support on a fixed subspace of dimension 
at most~$2^{4br}$, independent of the inputs, at any moment in the protocol.
This follows from an argument due to 
Yao~\cite{Yao93-quantum-circuit-complexity}, by considering a two party
protocol obtained by grouping all parties except~$\rA_i$ together (see Appendix~\ref{sec-yao93}). So
the state of party~$\rA_i$ at any point in the protocol can be mapped to 
one over~$4br$ qubits. Using this for the bound~$s$ on the local
memory, bandwidth~$b \in \Order( \log n)$, and the same reasoning as before, we 
get a lower bound of~$\widetilde{\Omega} (n d)^{1/3}$ on the number of 
rounds~$r$ in~$\Pi_d$. 
The precise expression for the bound with the logarithmic terms is
\[
\Omega \left( \frac{ (nd)^{1/3} }{ (\log n)^{1/3} \log^{8/3} \frac{ n}{ 
d^2 \log n } } \right) \enspace.
\]

We refine the analysis further to obtain Theorem~\ref{thm-lb},
by appealing to an
information-theoretic argument. The key insight is that regardless of
the size of the local memory maintained by the parties~$\rA_i$, for~$i
\in [d-1]$, 
the new \emph{information\/} they get about either input~$x$ or~$y$ in one 
round is bounded by~$b$, the length of the message from~$\rA_0$
or~$\rA_d$, respectively. Thus, the total information contained in the memory
and messages of these parties about the inputs may be bounded by~$rb$ at
any point in the protocol (see Lemma~\ref{lem-cil-pid}). This carries over to the information
contained in the messages between Alice and Bob in the two-party
protocol~$\Pi$ derived from~$\Pi_d $. The conditional information leakage of the
two-party protocol may then be bounded by~$rbm$, where~$m$ is the number
of rounds in~$\Pi$ (for suitable distributions over the inputs).
\begin{lemma}
\label{lem-cil-pi}
Let~$XYZ$ be jointly distributed random variables such that~$X, Y \in \set{0,1}^n$, and~$X$ and~$Y$ are independent given~$Z$. The conditional information leakage of the two-party protocol~$\Pi$ for Set Disjointness~$\disj_n$ mentioned in Lemma~\ref{lem-simulation} is bounded as~$\til( \Pi \,|\, XYZ) \in \Order(r^2 b / d)$.
\end{lemma}
We derive this as Corollary~\ref{cor-cil-pi} in Section~\ref{sec-cil}.

We now appeal to the following result due to Jain \etal~\cite{JRS03-set-disjointness} on the conditional information leakage of bounded-round protocols for Set Disjointness. This result is implicit in the proof of the~$\Omega(n / m^2)$ lower bound on the communication required by~$m$-round quantum protocols for Set Disjointness. (See Appendix~\ref{sec-jrs} for the details, including the significance of the auxiliary random variable~$Z$.)
\begin{theorem}[Jain, Radhakrishnan, Sen~\cite{JRS03-set-disjointness}]
\label{thm-jrs}
There is a choice of distribution for~$XYZ$ such that~$X, Y \in \set{0,1}^n$, the random variables~$X$ and~$Y$ are independent given~$Z$, and for any bounded-error entanglement-assisted two-party quantum communication protocol~$\Gamma$ for Set Disjointness~$\disj_n$ with~$m$ rounds, the conditional information leakage~$\til( \Gamma \,|\, XYZ)$ is at least~$\Omega(n/m)$.
\end{theorem}
Since the number of rounds~$m$ in the two-party protocol~$\Pi$ is at most~$2 \ceil{r/d}$, we conclude the~$\Omega( n d^2 / b)^{1/3}$ lower bound stated in Theorem~\ref{thm-lb}.

Corollary~\ref{cor-diam} follows by combining the lower bound for~$\linedisj_{n,d}$ and the
reduction from~$\linedisj_{n,d}$ to the Diameter Problem
described in the proof of Theorem~1.3 in Ref.~\cite{LM18-diameter-congest-model}.
In more detail, the reduction is based on a construction due to Abboud, Censor-Hillel, and Khoury~\cite{ACHK16-sparse-networks}. It involves a network~$G_{n,d}(x,y)$ in which the number of nodes is determined by~$n$ and~$d$, but the edges may also depend on the inputs~$x,y$ to~$\linedisj_{n,d}$. Given an instance of~$\linedisj_{n,d}$, the~$ d + 1$ parties~$(\rA_i)$ locally construct parts of the network~$G_{n,d}(x,y)$. The parts are such that each party holds a disjoint subset of the vertices of~$G_{n,d}(x,y)$, and each edge of~$G_{n,d}(x,y)$ is between vertices held either by adjacent parties or by the same party. Only the edges between the vertices held by~$\rA_0$ may depend on the input~$x$ given to it, and only the edges between the vertices held by~$\rA_d$ may depend on the input~$y$ given to it. The remaining edges are determined by~$n,d$. Given an algorithm for the Diameter Problem, the parties~$(\rA_i)$ are then able to simulate it on the graph~$G_{n,d}(x,y)$. In particular, these observations imply that shared entanglement in the network~$G_{n,d}(x,y)$ of the type described in Section~\ref{sec-congest} translates to shared entanglement between the parties~$(\rA_i)$ of the same type. We refer the reader to Refs.~\cite{ACHK16-sparse-networks,LM18-diameter-congest-model} for the remaining details of the reduction.

\subsection{Formal description of protocols on the line}
\label{sec-line-protocols}

In order to establish the lemmas stated in Section~\ref{sec-overview}, we
introduce some conventions and notation associated with multi-party protocols 
on the line of the sort we study for~$\linedisj_{n,d}$. By using unitary 
implementations of measurements, we assume that all the local
operations in the protocol, except the final measurement to obtain the outcome
of the protocol, are unitary. We also assume that the input
registers~$X$ with~$\rA_0$ and~$Y$ with $\rA_d$ are read-only. I.e., the
input registers may only be used as control registers during the protocol,
and are retained by the respective parties throughout.

For ease of exposition, we use subscripts on the registers held by
all the parties to implicitly specify the state of the register and
the party which last modified the state of the register. 
At the beginning of round~$t+1$, for~$t \in \set{ 0, 1, \dotsc, r - 1}$,
party~$\rA_0$ holds registers~$X A_{0, t} \, L_{1, t} $, party~$\rA_d$
holds registers~$R_{d - 1, t} \, A_{d, t} \, Y$, and for~$i \in [d - 1]$, 
party~$\rA_i$ holds registers~$R_{i - 1, t} \, A_{i, t} \, L_{i + 1, t}
$.  The registers~$L_{i,t}$ and~$R_{i,t} $, for~$i \in [0,d]$
and~$t \in [0,r]$, all have~$b$ qubits. Except in the first round, the 
first subscript at the beginning of the round, say~$i$, indicates that
party~$\rA_i$ held the register in the previous round, and sent the
register to the neighbour that holds it in the current round.

\begin{figure}[t]
\centering
\includegraphics[width=390pt]{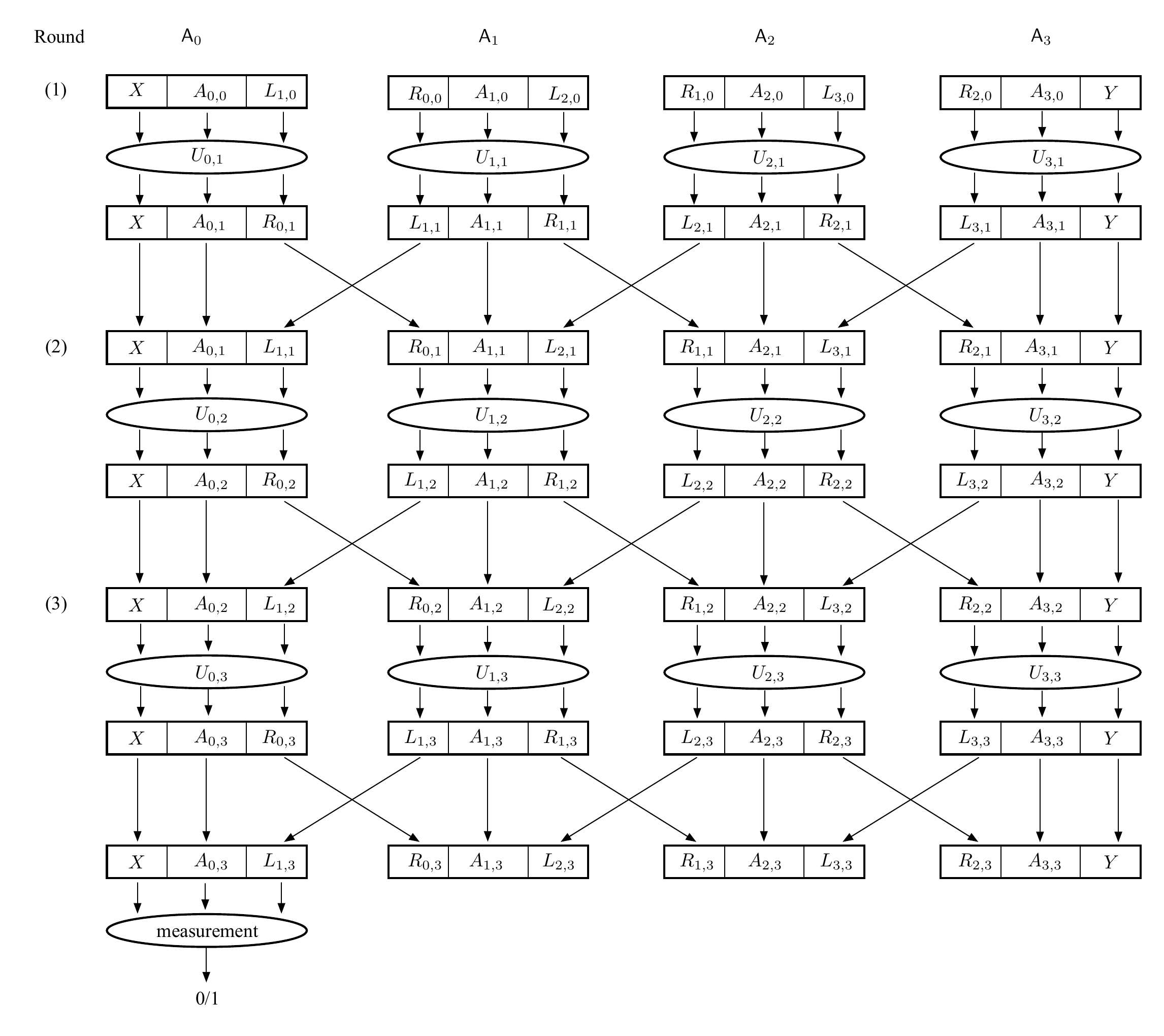}
\caption{A multi-party communication protocol on the line with~$4$ parties and~$3$ rounds, of the type we study for~$\linedisj_{n,d}$. For~$t \ge 1$, the subscripts~$i,t$ on a register indicate that the register was an ``output'' of the isometry applied by party~$\rA_i$ in round~$t$, and that it is in the corresponding state. For example, the register~$R_{1,3}$ was produced by the isometry applied by~$\rA_1$ in the third round.}
\label{fig-protocol}
\end{figure}

At the beginning of the first round, registers~$X$ and~$Y$ are
initialized to the input to the protocol. The qubits in the remaining 
registers are all initialised to a pure shared state that is
independent of the inputs. Note that this shared state also includes any
``work'' qubits in state~$\ancilla$ required by the parties.

In round~$t + 1$, each party~$\rA_i$ applies a unitary operation to the
registers they hold. We view the unitary operation as an isometry~$U_{i,
t + 1}$ that maps the registers to another sequence of registers with the
same dimensions. The registers~$X A_{0, t} \, L_{1, t}$ with~$\rA_0$
are mapped to~$X A_{0, t+1} \, R_{0, t+1}$. The registers~$R_{d - 1, t}
\, A_{d, t} \, Y$ with~$\rA_d$ are mapped to~$L_{d, t+1} \, A_{d,
t+1} \, Y$. For~$i \in [d-1]$, the registers~$R_{i - 1, t} \, A_{i, t}
\, L_{i + 1, t}$ with~$\rA_i$ are mapped to~$L_{i, t+1} \, A_{i, t+1} 
\, R_{i, t+1}$. So for~$t \ge 1$, the
subscripts~$i,t$ on a register indicate that the register was an ``output''
of the isometry applied by party~$\rA_i$ in round~$t$, and that it is in 
the corresponding state.

As the final action in round~$t+1$, for~$t < r$, if~$i > 0$, party~$\rA_i$ 
sends~$L_{i, t + 1}$ to the party on the left (i.e., to~$\rA_{ i - 1}$) 
and receives~$R_{i - 1, t + 1}$ from her; and if~$i < d$, she sends~$R_{i,
t + 1}$ to the party on the right (i.e., to~$\rA_{i + 1}$) and receives 
register~$L_{i + 1, t + 1}$ from her.

After the~$r$ rounds of the protocol have been completed, party~$\rA_0$
makes a two-outcome measurement, possibly depending on her input, on 
the registers~$A_{0, r} \, L_{1, r}$. The outcome is the output of the 
protocol. Figure~\ref{fig-protocol} depicts such a protocol.

\subsection{The two-party simulation}
\label{sec-simulation}

We now prove Lemma~\ref{lem-simulation}, by giving a formal description 
of the two-party protocol~$\Pi$ for Set Disjointness~$\disj_n$
derived from a protocol~$\Pi_d$ for~$\linedisj_{n,d}$. We use the 
notation and convention defined in Section~\ref{sec-line-protocols}
in our description below. For simplicity, we assume that the number of 
rounds~$r$ in~$\Pi_d$ is a multiple of~$d$, by adding dummy rounds with 
suitable local operations, if necessary. Since~$\disj_n$ depends 
non-trivially on \emph{both\/} inputs, the number of rounds~$r$ 
required to compute the function over a path of length~$d$ is at 
least~$d$. So the addition of dummy rounds may at most double the 
number of rounds.

\begin{figure}[h]
\centering
\includegraphics[width=350pt]{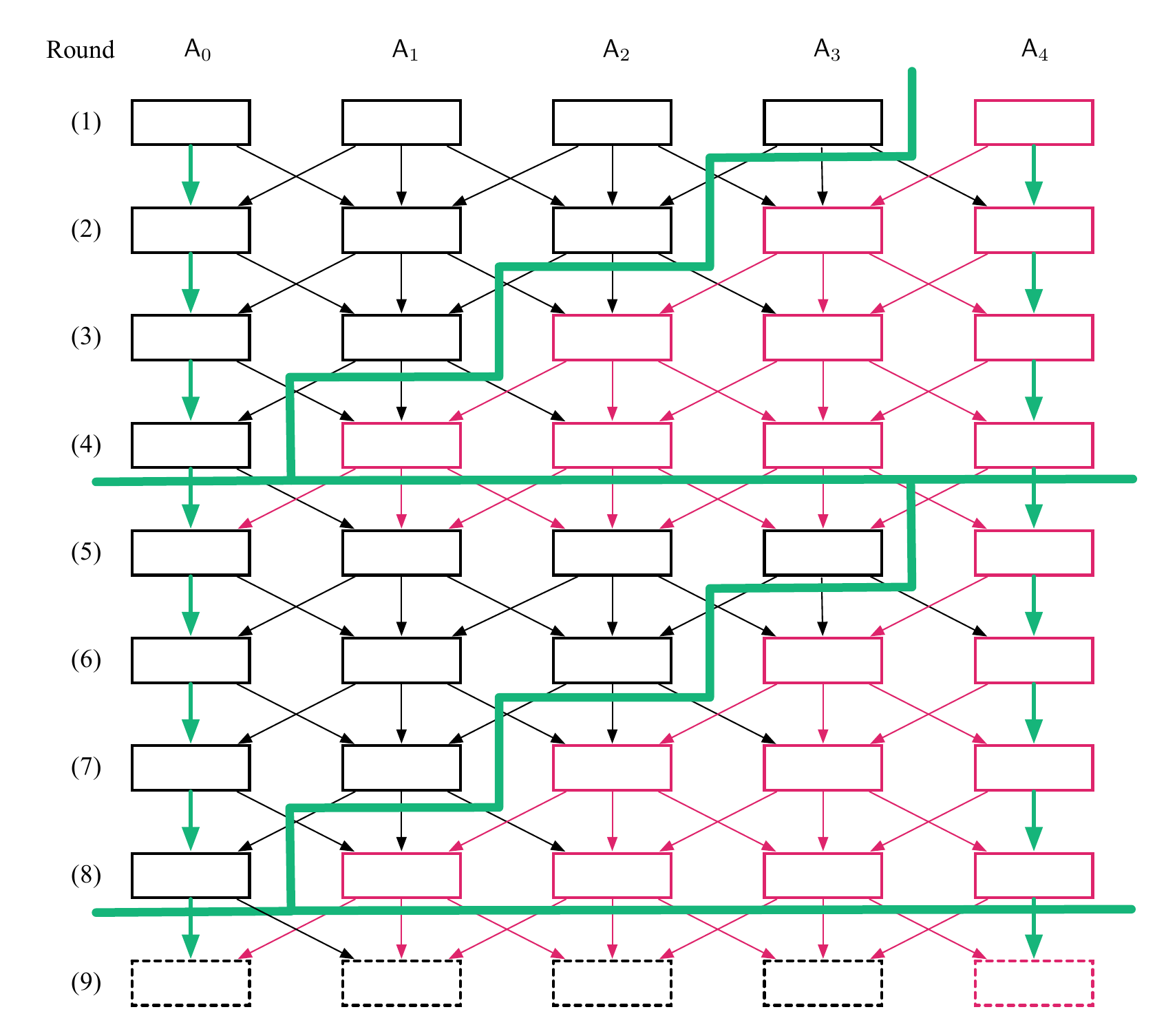}
\caption{A depiction of the two-party simulation of a multi-party communication protocol of the type we study for~$\linedisj_{n,d}$. Here, we have~$5$ parties and show the simulation of the first~$8$ rounds of the original protocol. Each round in the two-party protocol is delineated by thick green lines. The black rectangular boxes represent the isometries implemented by Alice, and the black arrows going across the thick green lines represent the communication from her to Bob. The red rectangular boxes represent the isometries implemented by Bob, and the red arrows going across the thick green lines represent the communication from him to Alice. The green arrows indicate that the input register and the local memory of the parties at the extremities are retained by them throughout.}
\label{fig-simulation}
\end{figure}

In the protocol~$\Pi$, Alice initially holds all the registers with
parties~$\rA_i$ for~$i < d$ at the beginning of the first round, and Bob 
holds the registers with~$\rA_d$. All of the registers are initialized 
as in~$\Pi_d$. The simulation implements blocks of~$d$ successive rounds 
of~$\Pi_d$ with two rounds in~$\Pi$, with Alice sending the message in
the first of the two rounds and Bob in the second. See Figure~\ref{fig-simulation} for a depiction of the simulation.

Assume that~$k$ blocks of~$d$ rounds each of~$\Pi_d$ have been 
implemented with~$2k$ rounds
in~$\Pi$, for some~$k \in [0, r/d - 1]$. We describe how the~$(k +
1)$-th block is implemented. Let~$t \eqdef kd$. We maintain the
invariant that at the beginning of the~$(2k + 1)$-th round in~$\Pi$,
Alice holds the registers~$X
A_{0, t} \, L_{1, t} $, and the registers~$R_{i - 1, t} \, A_{i, t} \,
L_{i + 1, t} $, for all~$i \in [d - 1]$. Alice's local operations in
round~$2k + 1$ are as follows. For each~$j \in \set{ t + 1, t + 2, 
t + 3, \dotsc, t + d }$ in increasing order (where~$j$ denotes a round 
in~$\Pi_d$), 
\begin{enumerate}

\item
Alice applies the isometry~$U_{0, j}$ to the registers~$X
A_{0, j - 1} \, L_{1, j - 1}$ to get registers~$X A_{0, j} \, R_{0, j} $.

\item 
For each~$l$ with~$1 \le l \le d - (j - t)$ (denoting a party 
from~$\Pi_d$), Alice applies the isometry~$U_{l, j}$ to the registers~$R_{l
- 1, j - 1} \, A_{l, j - 1} \, L_{l + 1, j - 1}$ to get registers~$L_{l,
j} \, A_{l, j} \, R_{l, j} $.

\item
For each~$l$ with~$1 \le l \le d - (j - t)$, Alice swaps
registers~$R_{ l - 1, j}$ and~$L_{ l, j} $.

\end{enumerate}
At this point, Alice has implemented the left upper triangular ``space-time 
slice'' of the~$(k + 1)$-th block of~$d$ rounds of~$\Pi_d$. She holds the 
registers
\begin{equation}
\label{eq-A-registers}
\begin{split}
X A_{0, t + d} \, R_{0, t + d} ~~ &
R_{0, t + d - 1} \, A_{1, t + d - 1} \, R_{1, t + d - 1} ~~
R_{1, t + d - 2} \, A_{2, t + d - 2} \, R_{2, t + d - 2} \\
R_{2, t + d - 3} \, A_{3, t + d - 3} \, R_{3, t + d - 3} & ~ 
\dotsb ~
R_{d - 3, t + 2} \, A_{d - 2, t + 2} \, R_{d - 2, t + 2} ~~
R_{d - 2, t + 1} \, A_{d - 1, t + 1} \, R_{d - 1, t + 1} \enspace,
\end{split}
\end{equation}
in the state implicitly specified by the subscripts. (The registers have
been grouped into threes, in the order of the parties that hold them in~$\Pi_d$.)
She sends all the registers \emph{except\/}~$X A_{0, t + d}$ to Bob.
This concludes the~$(2k + 1)$-th round of~$\Pi$.

We also maintain the invariant that at the beginning of the~$(2k + 2)$-th
round of~$\Pi$, Bob holds all the registers in Eq.~(\ref{eq-A-registers}) 
except~$X
A_{0, t + d} $, in addition to the registers~$R_{d - 1, t} \, A_{d, t}
Y$, where~$t = kd $. Bob's local operations in round~$2k + 2$ are as
follows. For each~$j \in \set{ t + 1, t + 2, t + 3, \dotsc, t + d }$ in
increasing order
(where~$j$ denotes a round in~$\Pi_d$ that Bob intends to complete),
\begin{enumerate}

\item
Bob applies the isometry~$U_{d, j}$ to the registers~$R_{d - 1, j - 1}
\, A_{d, j - 1} Y$ to get registers~$L_{d, j} \, A_{d, j} Y$.

\item
For each~$l$ with~$d - (j - t - 1) \le l \le d - 1$ (denoting a party
from~$\Pi_d$), Bob applies the isometry~$U_{l, j}$ to the
registers~$R_{l - 1, j - 1} \, A_{l, j - 1} \, L_{l + 1, j - 1}$ to get 
registers~$L_{l, j} \, A_{l, j} \, R_{l, j} $.

\item
For each~$l$ with~$d - (j - t - 1) \le l \le d$, Bob swaps
registers~$R_{l - 1, j}$ and~$L_{ l, j} $.

\end{enumerate}
At this point, Bob holds the registers
\begin{equation}
\label{eq-B-registers}
\begin{split}
L_{1, t + d} ~~ R_{0, t + d} \, A_{1, t + d} \, L_{2, t + d} ~~
R_{1, t + d} \, A_{2, t + d} & \, L_{3, t + d} ~~
R_{2, t + d} \, A_{3, t + d} \, L_{4, t + d} \\
\dotsb ~
R_{d - 2, t + d} \, A_{d - 1, t + d} \, L_{d, t + d} & ~~
R_{d - 1, t + d} \, A_{d, t + d} Y \enspace,
\end{split}
\end{equation}
in the state implicitly specified by the subscripts. The registers are
thus all in the state at the end of the~$(kd + d)$-th round in~$\Pi_d$.
Bob sends all the registers \emph{except\/}~$R_{d - 1, t + d} \, A_{d, t + d} Y$ to Alice.
This concludes the~$(2k + 2)$-th round of~$\Pi$, and the simulation of
the~$(k + 1)$-th block of rounds of~$\Pi_d$.

At the end of the simulation of the~$(r/d)$-th block of rounds
of~$\Pi_d$, Alice measures the registers~$A_{0, r} \, L_{1, r}$ as 
in~$\Pi_d$ to obtain the output. (As in~$\Pi_d$, this measurement may be 
controlled by the input register~$X$.) This completes the description of 
the two-party simulation. The 
correctness of the simulation follows by induction, by observing that 
Alice and Bob implement all the local operations and communication 
in~$\Pi_d$ in the correct order and with the correct registers.  
Lemma~\ref{lem-simulation} thus follows.

\subsection{Conditional information leakage of the two-party protocol} 
\label{sec-cil}

We are now ready to bound the conditional information leakage of the 
two-party protocol~$\Pi$ derived from the multi-party protocol~$\Pi_d$, 
with respect to a distribution~$\mu$ on the inputs and an auxiliary 
random variable. We initialise registers~$\hX \hY \hZ X Y Z $ to the 
canonical purification
\[
\sum_{x,y,z}  \sqrt{\mu(x,y,z)} \; \ket{xyz}^{\hX \hY \hZ} 
    \ket{xyz}^{X Y Z} \enspace,
\]
and run the protocol~$\Pi_d$ (and therefore~$\Pi$) on 
the input registers~$X$ and~$Y$, along with the other registers they
need. We use~$X, Y, Z$ to also refer to the input and auxiliary random 
variables. Suppose that~$\mu$ is such that~$X$ and~$Y$ are independent 
given~$Z$. We imagine that the purification register~$\hX$ is given to 
party~$\rA_0$ in~$\Pi_d$ (or to Alice in~$\Pi$), the register~$\hY$ is 
given to party~$\rA_d$ in~$\Pi_d$ (or to Bob in~$\Pi$), and the 
registers~$\hZ Z$ are held by a party not involved in either protocol.

We first bound the information contained about an input held by a party~$\rA_i$ ($i \in \set{0,d}$) in the 
registers held by all other parties~$\rA_j$, $j \neq i$, in~$\Pi_d$, conditioned on~$Z$.
For ease of notation, for~$t \ge 0$, we denote by~$D_t$ the entire
sequence of registers (including~$\hY$) held by the parties~$\rA_i$, with~$i \ge 1$, 
in the state at the end of the~$t$-th round of~$\Pi_d$. Similarly, 
we denote by~$C_t$ the entire sequence of registers (including~$\hX$) held by the 
parties~$\rA_i$, with~$i \le d - 1$, in the state at the end of 
the~$t$-th round of~$\Pi_d$.
\begin{lemma}
\label{lem-cil-pid}
For all~$t \ge 0$, we have~$\mi(X : D_t \,|\, Z) \le 2tb$, and~$\mi(Y :
C_t \,|\, Z) \le 2tb$.
\end{lemma}
\begin{proof}
We prove the bound on~$\mi(X : D_t \,|\, Z)$ by induction. 
The second bound is obtained similarly.

Let~$\sigma$ denote the state of the registers~$X D_t Z$, so that
\[
\sigma \quad = \quad \sum_z \lambda_z \; \sigma_z^{X D_t} \tensor \density{z}^Z \enspace,
\]
where~$\lambda$ is the marginal distribution of~$Z$, and~$\sigma_z$ is the state of the registers~$X D_t$, conditioned on the event~$Z = z$. By Lemma~\ref{lem-cq},
\[
\mi(X : D_t \,|\, Z)_\sigma \quad = \quad \expct_{z \sim \lambda} 
    \mi( X : D_t )_{\sigma_z} \enspace. 
\]
The base case~$t = 0$ is then immediate from the following observations.
The state of the registers of the parties~$\rA_i$, for~$i \in [d]$, except~$Y \hY$, 
is independent of the inputs, i.e., is in tensor product with the state of~$X Y \hY$.
Further, for every~$z$, the state of the registers~$Y \hY$ is in tensor product 
with that of~$X$, conditioned on~$Z = z$.

Assume that the bound holds for~$t = j$, with~$j \ge 0$. Let~$G_{j+1}$ 
denote the sequence of registers with all the parties~$\rA_i$, for~$i \ge 2$, 
after the isometry in round~$j + 1$ has been applied. Then we have
\[
\mi( X : L_{1, j+1}\, A_{1, j+1} \, R_{1, j+1} G_{j+1} \,|\, Z)
    \quad = \quad \mi(X : D_j \,|\, Z) \quad \le \quad 2jb \enspace,
\]
by the invariance of conditional mutual information under isometries and the induction hypothesis.
Let~$H_{j+1}$ denote all the registers of the parties~$\rA_i$, for~$i
\ge 2$, after the communication in round~$j + 1$. Then~$L_{2, j+1} \,
H_{j+1}$ and~$R_{1, j+1} \, G_{j+1}$ consist of the same set of registers,
but in different order.
By the properties of entropy and conditional mutual information mentioned 
below,
\begin{align*}
\mi( X : D_{j+1} \,|\, Z) 
    \quad & = \quad \mi( X : R_{0, j+1} \, A_{1, j+1} \, L_{2, j+1}
        \, H_{j+1} \,|\, Z) \\
    \quad & = \quad \mi( X : R_{0, j+1} \, A_{1, j+1} \, R_{1, j+1} \,
        G_{j+1} \,|\, Z) \\
    \quad & \le \quad 2 \entropy( R_{0, j+1} ) + \mi( X : A_{1, j+1} \,
        R_{1, j+1} \, G_{j+1} \,|\, Z) \\
    \quad & \le \quad 2b + \mi( X : L_{1, j+1}\, A_{1, j+1} \,
        R_{1, j+1} G_{j+1} \,|\, Z) \\
    \quad & \le \quad 2b + 2jb \enspace.
\end{align*}
The first inequality follows from Lemma~\ref{lem-mi-increase}, the second
by the property that~$\entropy(B)$ is bounded from above by the number of qubits in the 
register~$B$ and the data processing inequality (Lemma~\ref{lem-dpi}),
and the final one by the induction hypothesis.
\end{proof}

For~$l \in [2r/d]$,
denote the message registers in the~$l$-th round of the two-party 
protocol~$\Pi$ in the corresponding state together by~$M_l$. Denote the 
registers with Alice at the end of the~$l$-th round (including~$\hX$), in the
corresponding state, by~$E_l$, and the registers with Bob at the end of
the~$l$-th round (including~$\hY$), in the corresponding state, by~$F_l$.

Consider~$k \in [r/d]$. We observe from the definition of the protocol~$\Pi$, that for 
the odd numbered round~$2k - 1$, the state given by register~$D_{kd}$ is obtained by
an isometry on the registers~$M_{2k - 1} F_{2k - 2}$. The registers~$M_{2k - 1} 
F_{2k - 2}$ (in the state implicitly specified by their definition) are precisely 
the registers Bob holds at the end of round~$2k - 1$ of~$\Pi$. Moreover, for the
even numbered round~$2k$, the state given by the registers~$E_{2k - 1} M_{2k}$
is precisely the state given by the register~$C_{kd}$ in~$\Pi_d$. The 
registers~$E_{2k - 1} M_{2k}$ are precisely the registers Alice holds at the 
end of round~$2k$ of~$\Pi$. Therefore, by Lemma~\ref{lem-cil-pid} and the 
definition of conditional information leakage, we have:
\begin{corollary}
\label{cor-cil-pi}
For all~$k \in [r/d]$, we have
\begin{align*}
\mi( X : M_{2k - 1} F_{2k - 2} \,|\, Z)
    \quad & = \quad \mi( X : D_{kd} \,|\, Z) \quad \le \quad 2kdb
        \enspace, \qquad \text{and} \\
\mi( Y : E_{2k - 1} M_{2k} \,|\, Z)
    \quad & = \quad \mi( Y : C_{kd} \,|\, Z) \quad \le \quad 2kdb
        \enspace.
\end{align*}
Consequently, the conditional information leakage of~$\Pi$ is bounded
as~$\til( \Pi \,|\, XYZ) \le 4 r^2 b/ d $.
\end{corollary}

\section{Two-oracle query algorithms with a switching delay}
\label{sec-d-query-algorithms}

In this section, we define a new model of query complexity, two-oracle query complexity with a ``switching delay'', motivated by the study of Set Disjointness on a Line~$\linedisj_{n,d}$. The lower bound technique involving conditional information leakage that we use to establish Theorem~\ref{thm-lb} extends to the analogue of Set Disjointness~$\disj_n$ in this model, with a switching delay of~$d$ queries. As a consequence, it yields the same lower bound on \emph{query\/} complexity. Furthermore, we design a quantum algorithm that matches this bound up to a polylogarithmic factor. This shows that the lower bound on conditional information leakage for~$\disj_n$ stated in Theorem~\ref{thm-jrs} is optimal up to a logarithmic factor. Therefore, if the lower bound for~$\linedisj_{n,d}$ is not optimal, we would require different ideas to improve it. Due to the similarities between the two models of computation, the algorithm also hints at the possibility of a more efficient distributed algorithm for Set Disjointness on a Line.

\subsection{The new query model}
\label{sec-d-query-model}

Turning to the definition of the query model, we consider query algorithms for computing bi-variate functions~$f : \set{0,1}^n \times \set{0,1}^n \rightarrow \set{0,1}$. We define the quantum version of the model; the classical versions may be defined analogously. The inputs~$x,y$ to the algorithm are provided indirectly, through oracles~$\cO_x$ and~$\cO_y$, as defined in Section~\ref{sec-query-complexity}. The query algorithm is defined in the standard manner, as an alternating sequence of unitary operators independent of the inputs~$x,y$, and queries~$\cO_x$ or~$\cO_y$, applied to a fixed initial state (that is also independent of the inputs). Thus, the sequence of queries to the inputs is pre-determined. However, we define the complexity of the algorithm differently. In addition to the queries, we charge the algorithm for switching between a query to~$x$ and a query to~$y$. We include a cost of~$d$ in the complexity whenever the algorithm switches between a query to~$x$ and a query to~$y$. This cost parallels the cost of accessing the inputs in the distributed computing scenario in which the inputs are physically separated by distance~$d$.

We may simplify the above model as follows, at the expense of increasing the complexity by a factor of at most~$2$. In the simplified model, we require that the queries be made in \emph{rounds\/}. In each round, the algorithm makes~$d$ queries, but exclusively to one of the inputs~$x$ or~$y$. Further, the algorithm alternates between the two oracles~$\cO_x$ and~$\cO_y$ in successive rounds. The complexity of the algorithm is now defined in the standard manner, as the total number of queries in the algorithm. Thus the complexity equals~$d$ times the number of rounds.

It is straightforward to verify that any algorithm with complexity~$q$ in the first model has complexity at most~$2q$ in the second model, for computing any function~$f$ that depends on both inputs (i.e., when the algorithm in the first model queries both oracles). Furthermore, any algorithm with complexity~$q$ in the second model has complexity at most~$2q$ in the first model. The two models are thus equivalent up to a factor of two in complexity.

The second model is also relevant in a ``semi-parallel'' scenario, where a sequence of~$d$ queries are made to~$x$ independently of the answers to~$d$ other queries made to~$y$ during the same time steps. Up to a factor of~$2$ in complexity, this semi-parallel model can be simulated by the second model above. We thus adopt the second model in the definition below.

\begin{definition}
A \emph{two-oracle delay-$d$ quantum query algorithm\/} is a query algorithm~$\mathcal{A}$ with (predetermined) access to two oracles~$\cO_1, \cO_2$, which may be decomposed into some number of contiguous sequences of unitary operators called rounds such that each round contains~$d$ queries to the same oracle, and the algorithm alternates between the two oracles in successive rounds. 
The \emph{round complexity\/} of~$\mathcal{A}$ is the number~$r$ of rounds in a decomposition of~$\mathcal{A}$ as above.
The \emph{delay-$d$ query complexity\/} of $\mathcal{A}$ is $d\times r$.
\end{definition}

We define the \emph{quantum two-oracle delay-$d$ round complexity\/} of a bi-variate function~$f$ as the minimum round-complexity of any two-oracle delay-$d$ quantum query algorithm computing~$f$ with probability of error at most~$1/3$, given oracles~$\cO_x, \cO_y$ for the inputs~$x,y$. We define the \emph{quantum two-oracle delay-$d$ query complexity\/} of~$f$ similarly. We may assume that~$d\leq n$, as otherwise, an algorithm can learn~$x$ and~$y$ in two rounds.

Adapting the tools developed in Section~\ref{sec-linedisj} we get the following lower bound.
\begin{theorem}
\label{thm-lb2}
Let $d\leq n$.
The quantum two-oracle delay-$d$ round complexity of Set Disjointness~$\disj_n$ is 
$\Omega(\sqrt{n}/d)$ and $\Omega(\sqrt[3]{n/(d \log n)} \,)$.
The quantum two-oracle delay-$d$ query complexity is $\Omega(\sqrt{n} \,)$ and $\Omega(\sqrt[3]{n d^2 / \log n} \,)$.
\end{theorem}
Note that the first expression for either bound dominates when~$d^4 \in \Order( n \log^2 n)$. 

We briefly sketch the proof of Theorem~\ref{thm-lb2}. The query lower bound follows from the one on rounds. The~$\Omega( \sqrt{n}/ d)$ lower bound on rounds follows by observing that Set Disjointness~$\disj_n$ simplifies to the unordered search problem (OR function on~$n$ bits) in the standard quantum query model when we set~$y$ to be the all~$1$s string. For the second lower bound, we view a query to an oracle~$\cO_x$ or~$\cO_y$ as the exchange of~$2(\log n + 1)$ qubits between the algorithm and the oracle. So we can use any~$r$-round algorithm for computing~$f$ in the two-oracle delay-$d$ query model to derive a two-party communication protocol for computing~$f$ also with~$r$ rounds. The two parties run the query algorithm, each party sending all its registers to the corresponding player, whenever the algorithm switches between queries to~$x$ and queries to~$y$. In each round, the state of the algorithm (therefore the corresponding message) accumulates at most~$2d(\log n + 1)$ qubits of \emph{additional\/} information about either input. This is a consequence of the same kind of reasoning as in Lemma~\ref{lem-cil-pid}. Thus the conditional information leakage of the resulting two-party protocol may be bounded by~$2 r^2 d(\log n + 1)$. By Theorem~\ref{thm-jrs}, this is~$\Omega(n / r)$, so we get the~$\Omega(\sqrt[3]{n/(d \log n)} \,)$ lower bound for the number of rounds stated in Theorem~\ref{thm-lb2}.

\subsection{Algorithm for Set Disjointness}

Finally, we present an algorithm in the two-oracle model that matches the lower bounds stated in Theorem~\ref{thm-lb2}, up to polylogarithmic factors.
\begin{theorem}
\label{thm-algo}
Let $d\leq n$. The quantum two-oracle delay-$d$ round and query complexity of Set Disjointness~$\disj_n$ are 
\begin{itemize}
    \item $\Order(\sqrt{n \log n}/d)$ and $\Order(\sqrt{n \log n} \,)$, respectively, when~$d^4 \leq n \log^3 n$; and
    \item $\Order(\sqrt[3]{n/d} \,)$ and $\Order(\sqrt[3]{n d^2} \,)$, respectively, when~$d^4 \ge n \log^3 n$.
\end{itemize}
\end{theorem}
\begin{proof}
We present a quantum two-oracle delay-$d$ query algorithm with a parameter~$t \in [n]$, which gives the round and query bounds for suitable choices of~$t$ depending on how large~$d$ is as compared with~$n$.

\suppress{
The quantum algorithm runs in two stages. First, it 
searches for a subset~$I \subseteq [n]$ of size~$t$ such that it contains an index~$i\in [n]$ with~$x_i=y_i=1$. If it succeeds in finding such a subset~$I$, in the second stage, the algorithm looks for an index~$i\in I$ such that~$x_i = y_i = 1$. For this, it sequentially runs through the indices in~$I$ and checks if the requisite condition is satisfied. The second stage can thus be implemented in $\Order(\max \set{ 1, t/d } )$ rounds. The choice of~$t$ is such that the number of rounds in the first stage always dominates, and gives us the stated bounds. 
}

The quantum algorithm
searches for a subset~$I \subseteq [n]$ of size~$t$ such that it contains an index~$i\in [n]$ with~$x_i=y_i=1$. If it succeeds in finding such a subset~$I$, we may also find an index~$i\in [n] $ with~$x_i=y_i=1$ without increasing the asymptotic complexity of the algorithm (although this is not required for computing~$\disj_n$). For this, the algorithm sequentially runs through the indices in~$I$ and checks if the requisite condition is satisfied. This second stage of the algorithm can thus be implemented in $\Order(\max \set{ 1, t/d } )$ rounds. The choice of~$t$ is such that the number of rounds in the first stage always dominates, and gives us the stated bounds. 

We describe the first stage next. In order to identify a subset~$I$ containing an index~$i$ as above, if there is any, we implement a search algorithm based on a quantum walk on the Johnson Graph~$J(n,t)$, following the framework due to Magniez, Nayak, Roland, and Santha~\cite{mnrs11}. The vertices of~$J(n,t)$ are $t$-subsets of~$[n]$. There is an edge between two vertices~$I, I'$ in~$J(n,t)$ iff~$I$ and~$I'$ differ in exactly~$2$ elements: $(I \setminus I') \union (I' \setminus I) = \set{i,j}$ for distinct elements~$i,j \in [n]$.

The three building blocks of such an algorithm are as follows.

\begin{description}
    \item[Set-up:] Construct the following starting superposition:
    $$ \binom{n}{ t}^{-1/2} \sum_{I\subseteq [n] \;:\; |I|=t} \ket{(i,x_i) : i \in I}  \enspace.$$
    
   \item[Checking:] Check whether $x_i=y_i=1$ for some $i\in I$:
$$
\ket{(i,x_i) : i \in I} \quad \mapsto \quad 
\left\{
%\begin{cases}
\begin{array}{rl}
-\ket{(i,x_i) : i \in I}, & \text{if $x_i=y_i=1$ for some $i\in I$} \enspace; \\
\ket{(i,x_i) : i \in I}, & \text{otherwise.}
\end{array}
%\end{cases}
\right.
$$
\item[Update:] Replace some index~$j \in I$ by an index~$k\not\in I$, and update the corresponding bit~$x_j$ to~$x_k$:
$$
\ket{(i,x_i) : i \in I}\ket{j}\ket{k} \quad \mapsto \quad \ket{(i,x_i) : i \in (I\setminus\{j\})\cup\{k\}}\ket{k}\ket{j}  \enspace.
$$
\end{description}

Let $\varepsilon$ be the probability that a uniformly random $t$-subset of~$[n]$ contains an index~$i$ such that $x_i = y_i = 1$, given that such an element~$i$ exists. We have~$\varepsilon \in \Omega(t/n)$. 
Then, according to Theorem~1.4 in Ref.~\cite{mnrs11}, there is an algorithm based on quantum walk that finds a subset~$I$ such that $x_i=y_i=1$ for some $i\in I$, if there is any such subset, with constant probability~$> 1/2$. 
The algorithm uses one instance of \textbf{Set-up}, and~$\Order( \sqrt{1/\varepsilon} \,)$ alternations of one instance of \textbf{Checking} with a sequence of~$\Order(\sqrt{t} \,)$ instances of \textbf{Update},  interspersed with other unitary operations that are independent of the inputs~$x,y$. (The spectral gap of the Johnson graph needed in the analysis of the algorithm may be derived from the results in Ref.~\cite{K93-combinatorial-matrices}, for example.)

Note that \textbf{Set-up} uses~$t$ queries to $x$, and thus can be implemented in $\max(1,2 \ceil{t/d})$ rounds. \textbf{Update} only requires $2$ queries to $x$. Thus a sequence of $\sqrt{t}$ sequential \textbf{Update} operations can be implemented in order~$\max(1,2\sqrt{t}/d)$ rounds. 
We would like to use the Grover algorithm for unordered search to implement the checking step. The Grover algorithm incurs non-zero probability of error in general, while the algorithm due to Magniez \etal\ assumes that the checking step is perfect. We therefore use an algorithm for unordered search with small error due to Buhrman, Cleve, de Wolf, and Zalka~\cite{BCdWZ99-small-error-search} to implement \textbf{Checking} with error at most~$c \sqrt{t/n}$ for a suitable positive constant~$c$ with order~$\sqrt{t \log(n / t) }$ queries to $y$. Using standard arguments, this only increases the error of the quantum walk algorithm by a small constant, say~$1/10$. In effect, \textbf{Checking} (with the stated error) can be implemented in order~$\max(1,\sqrt{t \log n}/d)$ rounds. Thus the bound on the round complexity of the quantum walk algorithm is of the order of
\begin{equation}
\label{eq-rounds}
\max \set{1, \frac{t}{d} } + \sqrt{ \frac{n}{t}} \left( \max \set{ 1, \frac{ \sqrt{t \log n} }{ d} } 
    + \max \set{ 1, \frac{ \sqrt{t}}{d} } \right) \enspace.
\end{equation}

In order to derive the bounds stated in the theorem, we optimise over~$t$. We consider intervals of values for~$t$ such that each of the expressions involving maximisation in Eq.~(\ref{eq-rounds}) simplifies to one of the terms. The intervals are given by partitioning~$[n]$ at the points~$d, d^2/\log n, d^2$. (Note that~$d$ need not be smaller than~$d^2 / \log n$.) We optimise the number of rounds within each interval, which in turn gives us a relation between~$d$ and~$n$ for which the rounds are minimised. 

We first consider~$d \le \log n$, so that~$d^2 / \log n \le d$, and~$t$ in the intervals
\[
[1, d^2 /\log n], \quad [d^2 / \log n, d], \quad [d, d^2], \quad \text{and}~[d^2, n] \enspace.
\]
We optimise the number of rounds with~$t$ in each of these intervals, to find that the number of rounds is~$\Order( \sqrt{n \log n} / d)$ when~$t \eqdef d$. The optimal values of~$t$ in the other intervals also give the same bound, but we stay with~$t = d$ so as to minimise the rounds in the (optional) second stage of the algorithm.

Next we consider~$d \ge \log n$, so that~$d \le d^2 / \log n \le d^2$. We again optimise over~$t$ in four intervals, and get the following bounds:
\begin{enumerate}
    \item $t \in [1, d]$: $\Order( \sqrt{n / d} \,)$ when~$t \eqdef d$.
    \item $t \in [d, d^2 / \log n]$: $\Order( \sqrt[3]{n/d} \,)$ when~$t \eqdef \sqrt[3]{n d^2}$, \emph{provided\/}~$d^4 \ge n \log^3 n$. If~$d^4 \le n \log^3 n$, we get~$\Order( \sqrt{n \log n}/ d )$ when~$t \eqdef d^2 / \log n$.
    \item $t \in [d^2 / \log n, d^2]$: $\Order( \sqrt{n \log n}/d)$ when~$t \eqdef d^2 / \log n$ \emph{provided\/}~$d^4 \le n \log^3 n$. If~$d^4 \ge n \log^3 n$, we get~$\Order(d / \log n)$ with the same value of~$t$.
    \item $t \in [d^2, n]$: $\Order( \sqrt{n \log n}/d)$ when~$t \eqdef d^2$ \emph{provided\/}~$d^4 \le n \log n$. If~$d^4 \ge n \log n$, we get~$\Order(d)$ with the same value of~$t$.
\end{enumerate}
Since~$\sqrt{n/d} \ge \sqrt{n \log n}/ d$ when~$d \ge \log n$, $\sqrt{n \log n}/ d \le d$ when~$d^4 \ge n \log n$, and~$(n/d)^{1/3} \le d/ \log n$ when~$ d^4 \ge n \log^3 n $, we conclude the bounds on round complexity stated in the theorem:
\begin{itemize}
    \item $\Order( \sqrt{n \log n} / d)$ with~$t \eqdef d$ when~$d \le \log n$, or with~$t \eqdef d^2 / \log n$ when~$\log^4 n \le d^4 \le n \log^3 n \,$, and
    \item $\Order( \sqrt[3]{n/d} \,)$ with~$t \eqdef \sqrt[3]{n d^2}$ when~$d^4 \ge n \log^3 n \,$.
\end{itemize}
The bounds on query complexity follow.
\end{proof}

Note that in the range of parameters such that~$n \log^2 n \le d^4 \le n \log^3 n$, the upper bound~$\sqrt{ n \log n}/ d$ is at most~$\sqrt{\log n}$ times the lower bound~$\sqrt[3]{n / d \log n} \,$. So the bounds in Theorems~\ref{thm-lb2} and~\ref{thm-algo} are indeed within polylogarithmic factors of each other for all values of~$d,n$ (such that~$d \le n$).

\section{Conclusion}
\label{sec-concl}

In this work, we studied a fundamental problem, Set Disjointness, in two concrete computational models. Set Disjointness on the Line~$\linedisj_{n,d}$ reveals new subtleties in distributed computation with quantum resources. It again puts the spotlight on the ``double counting'' of information in conditional information loss (and leakage). One may think that the more sophisticated notion of \emph{quantum information cost\/} introduced by Touchette~\cite{Touchette15-QIC}, along with the results due to Braverman \etal~\cite{BGKMT18-bounded-round-disj}, might help us overcome this drawback. Indeed, quantum information cost helps us overcome the limitations of the former quantity in the case of Set Disjointness in the standard two-party communication model. Surprisingly, these techniques do not seem to help in obtaining a better lower bound for~$\linedisj_{n,d}$. (An analysis of the quantum information cost of the two-party protocol derived in Lemma~\ref{lem-simulation}, along with the lower bound on this quantity shown by Ref.~\cite{BGKMT18-bounded-round-disj}, gives us a bound that is a poly-logarithmic factor smaller than the one we derive in Theorem~\ref{thm-lb}.) We believe that new ideas may be needed to characterise its asymptotic round complexity. 

The two-oracle query model we introduce gives us a different perspective on Set Disjointness on a Line. It implies that any improvement to the round lower bound  for~$\linedisj_{n,d}$ would necessarily go beyond the use of conditional information leakage for two-party protocols for Set Disjointness. The algorithm also suggests that more efficient protocols for~$\linedisj_{n,d}$ may exist. More generally, the new query model is tailored towards the study of distributed algorithms on the line and could shed light on protocols for other similar problems. Moreover, the model could also be of relevance in other distributed computation scenarios.

\bibliography{refs}

\appendix

\section{Conditional information leakage of Set disjointness}
\label{sec-jrs}

Theorem~\ref{thm-jrs}, the lower bound on the conditional information leakage of bounded-round protocols for Set Disjointness due to Jain, Radhakrishnan, and Sen~\cite{JRS03-set-disjointness} is not stated explicitly in their article. In this section, we explain how the theorem may be inferred from their work.

Jain \etal\ implicitly analyse the conditional information leakage of protocols for Set Disjointness and relate it to the conditional information loss of a protocol for the two-bit AND function. For completeness, we define the latter quantity using the same notation as for the conditional information leakage of a quantum communication protocol~$\Pi$ introduced in Sec.~\ref{sec-qcc}. 

The \emph{conditional information loss\/}~$\il( \Pi \,|\, XYZ)$ of the protocol~$\Pi$ is defined as
\[
\il( \Pi \,|\, XYZ) \quad \eqdef \quad \sum_{i \in [m], ~ i \text{ odd}} \mi( X : B_i Y \,|\, Z) 
    + \sum_{i \in [m], ~ i \text{ even}} \mi( Y : A_i X \,|\, Z) \enspace,
\]
where the registers are implicitly assumed to be in the state given by the protocol. The difference between conditional information leakage and loss lies in the inclusion of the purification registers~$\hX$ and~$\hY$ in the mutual information terms in the former quantity. By the Data Processing Inequality (Lemma~\ref{lem-dpi}), the conditional information leakage of a protocol is at least as large as its conditional information loss.

Jain \etal\ prove an~$\Omega(n / m^2)$ lower bound on the communication required by any~$m$-round entanglement-assisted two-party quantum communication protocol~$\Gamma$ for Set Disjointness. They show this in three steps. In the first step, they bound the conditional information leakage~$\til(\Gamma| XYZ)$ by~$2 m c$, where~$c$ is the total number of qubits exchanged in~$\Gamma$, and~$XYZ$ are any jointly distributed random variables such that~$X, Y \in \set{0, 1}^n$ and~$X$ and~$Y$ are independent given~$Z$. In the second step, they show that there is a specific distribution for~$XYZ$ with the properties stated above, a protocol~$\Gamma'$ for the two-bit AND function (derived from~$\Gamma$), and jointly distributed random variables~$X' Y' Z'$ with~$X', Y' \in \set{0,1}$ such that conditional information \emph{loss\/}~$\il(\Gamma'|X'Y'Z')$ is at most~$\til(\Gamma| XYZ) /n$. In the third step, they show that~$\il(\Gamma'|X'Y'Z')$ for AND is at least $\Omega(1/m)$. Theorem~\ref{thm-jrs} follows by combining the last two steps. The auxiliary random variable~$Z$ and the conditioning on this random variable are what enable the \emph{direct sum\/} property underlying the reduction in the second step.

The purification registers used in conditional information leakage are required so that the joint state of the two parties in the protocol~$\Gamma'$ for the AND function derived from~$\Gamma$ on any fixed input is pure until the measurement used for producing the output. This property is crucial for the third step of the proof described above. (It turns out, though, that for the distribution~$ X' Y' Z'$ they derive, the conditional information leakage of~$\Gamma'$  coincides with its conditional information loss.)

Jain \etal\ present their communication lower bound in more generality, for~$t$-party protocols, for~$t \ge 2$, and for a class of functions that includes Set Disjointness. The first two steps are proven together in Lemma~2 in Ref.~\cite{JRS03-set-disjointness}, and the first step can be inferred from the derivation of Eq.~(1) in the proof. The third step is proven in Lemma~3.

Theorem~\ref{thm-jrs} may be easier to infer from the arXiv pre-print~\cite{JRS03-arxiv-set-disjointness}, as this version concerns two-party protocols. Conditional information loss is presented in Definition~5, the first two steps described above are proven together in Lemma~3, and the third step is proven in Lemma~4.

\section{States in two-party communication protocols}
\label{sec-yao93}

As it appears not to be well-known, we include a result on the structure of the joint states in a two-party quantum communication protocol \emph{without\/} shared entanglement, and a consequence of relevance to us. We state the result in the notation introduced in Section~\ref{sec-qcc}.

\begin{lemma}[Yao~\cite{Yao93-quantum-circuit-complexity}]
\label{lem-yao93}
Consider a two-party quantum communication protocol without shared entanglement in which Alice gets input~$x$ and Bob gets input~$y$. The joint state~$\ket{\psi_k (x,y)}$ of the work registers~$A_k B_k$ at the end of the~$k$-th round of the protocol may be expressed as
\[
\ket{\psi_k (x,y)} \quad = \quad \sum_{c \, \in \, \set{0,1}^{q_k}} \ket{ \phi(x,c) }^{A_k} \tensor \ket{ \xi(y,c) }^{B_k} \enspace, 
\]
where~$q_k $ is the sum of the lengths of the first~$k$ messages, and~$\ket{ \phi(x,c) }$ and~$\ket{ \xi(y,c) }$ are possibly non-normalised states of appropriate dimension that depend only on~$x,c$ and~$y,c$, respectively.
\end{lemma}
\begin{proof}
We prove the statement by induction over~$k$. For~$k = 0$, we have~$q_k \eqdef 0$, and we may write the state of~$A_0 B_0$ as
\[
\ket{\psi_0 (x,y)} \quad = \quad \ancilla^{A_0} \tensor \ancilla^{B_0} \enspace.
\]
Assume that the statement holds for~$k = j$, for some~$j \ge 0$, so that
\[
\ket{\psi_j (x,y)} \quad = \quad \sum_{c \, \in \,  \set{0,1}^{q_j}} \ket{ \phi(x,c) }^{A_j} \tensor \ket{ \xi(y,c) }^{B_j} \enspace, 
\]
with~$q_j$, $\ket{ \phi(x,c) }$, and~$\ket{ \xi(y,c) }$ as in the statement of the lemma.

Consider round~$j + 1$. Suppose that Alice sends the message in the~$(j+1)$-th round; the other case is analogous. Suppose Alice applies the isometry~$U_x$ (depending on her input~$x$) to the work register~$A_j$ to obtain registers~$A_{j+1} M_{j+1}$ in state~$\ket{ \phi'(x,c) } \eqdef U_x \ket{ \phi(x,c) }$. She then sends the message register~$M_{j+1}$ to Bob. Bob's work register at the end of the~$(j+1)$-th round is then~$B_{j+1} \eqdef M_{j+1} B_j$. Suppose~$M_{j+1}$ consists of~$q$ qubits. For each pair~$x,c$, We may express the state~$\ket{ \phi'(x,c) }$ as
\[
\ket{ \phi'(x,c) } \quad = \quad \sum_{c' \, \in \, \set{0,1}^{q}} \ket{ \phi(x,c c') }^{A_{j+1}} \tensor \ket{ c' }^{M_{j+1}} \enspace,
\]
for suitable non-normalised states~$\ket{ \phi(x,c c') }$ of appropriate dimension.
So we may write the state~$\ket{\psi_{j+1} (x,y)}$ as
\[
\ket{\psi_{j+1} (x,y)} \quad = \quad \sum_{c \, \in \, \set{0,1}^{q_j}} \sum_{c' \, \in \, \set{0,1}^{q}} \ket{ \phi(x,c c') }^{A_{j+1}} \tensor \ket{ \xi(y,c c') }^{B_{j+1}} \enspace, 
\]
where
\[
\ket{ \xi(y,c c') } \quad \eqdef \quad \ket{ c' }^{M_{j+1}} \tensor \ket{ \xi(y,c) }^{B_j} \enspace.
\]
This proves the lemma.
\end{proof}

The lemma implies that the state of the register~$B_m$ at the end of an~$m$-round protocol has support in the linear span of the states~$ \set{ \ket{ \xi(y,c) } : c \in \set{0,1}^m }$, when the input given to Bob is~$y$. When Bob has no input, the support has dimension at most~$2^{q_m}$, independent of Alice's input~$x$. We may thus define an isometry~$W$ such that~$W \ket{ \xi(c) } =  \ket{ \xi'(c) } \tensor \ancilla$, where~$\ket{ \xi'(c) }$ is a~$q_m$-qubit state. Effectively, we need at most~$q_m$ qubits to store the final state. This argument may be extended to all the previous rounds, i.e., we may define suitable isometries to store the states in each round in~$q_m$ qubits. Let~$W_j$ be the isometry used for this purpose at the end of the~$j$-th round, for~$j \ge 0$. We may also modify the isometry~$V_j$ applied by Bob in the~$j$-th round to~$V'_j \eqdef W_j V_j W_{j-1}^\adjoint$. If Bob performs the measurement to produce the output of the protocol, the measurement may be modified similarly. Thus, we get a protocol in which Bob uses at most~$q_m$ work qubits throughout. 

In the context of the second paragraph after Lemma~\ref{lem-simulation}, for a fixed~$i \in [d-1]$, suppose Bob simulates the actions of party~$\rA_i$, while Alice simulates the actions of all the other parties. Then the protocol~$\Pi_d$ translates to a two-party protocol with~$2r$ rounds and total communication of~$4rb$. The isometries used by~$\rA_i$ are all independent of the inputs to~$\Pi_d$. The claim made in the said paragraph then follows.

\end{document}